\documentclass{article}
\usepackage{fullpage}
\usepackage{url}
\usepackage{varwidth}
\usepackage{graphics}
\usepackage[dvips]{epsfig}

\usepackage{amsmath}
\usepackage{amssymb}
\usepackage{amsfonts}
\usepackage{graphicx}
\usepackage{algorithm, algorithmic}

\newtheorem{theorem}{Theorem}[section]

\newtheorem{corollary}{Corollary}[section]

\newtheorem{proposition}{Proposition}[section]

\newtheorem{fact}{Fact}[section]

\newcommand{\qed}{\hfill $\Box$ \bigbreak}
\newenvironment{proof}{\noindent {\bf Proof.}}{\qed}

\newcommand{\remove}[1]{}

\begin{document}
\sloppy

\title{Time Versus Cost Tradeoffs for Deterministic\\ Rendezvous in Networks}
\author{Avery Miller, Andrzej Pelc\\
Universit\'{e} du Qu\'{e}bec en Outaouais\\
\url{avery@averymiller.ca}, \url{andrzej.pelc@uqo.ca}}




\maketitle

\begin{abstract}

Two mobile agents, starting from different nodes of a network at possibly different times, have to meet at the same node.
This problem is known as {\em rendezvous}.
Agents move in synchronous rounds.
Each agent has a distinct integer label from the set $\{1,\dots,L\}$.

Two main efficiency measures of  rendezvous are its {\em time} (the number of rounds until the meeting)
and its {\em cost} (the total number of edge traversals). We investigate tradeoffs between these two measures.
A natural benchmark for both time and cost of rendezvous in a network is the number of edge traversals needed for visiting all nodes of the network,
called the exploration time. Hence we express the time and cost 
of rendezvous as functions of an upper bound $E$ on the time of exploration (where $E$ and a corresponding exploration procedure are known to both agents) and of the size $L$ of the label space.
We present two natural rendezvous algorithms. Algorithm {\tt Cheap} has cost $O(E)$ (and, in fact, a version of this algorithm for  
 the model where the agents start simultaneously has cost exactly $E$) and time $O(EL)$. Algorithm {\tt Fast} has both time and cost $O(E\log L)$. Our main contributions are lower bounds showing that,
perhaps surprisingly, these two algorithms capture the tradeoffs between time and cost of rendezvous almost tightly. We show that
any deterministic rendezvous algorithm of cost asymptotically $E$ (i.e., of cost $E+o(E)$) must have time $\Omega(EL)$. On the other hand,
we show that any deterministic rendezvous algorithm with time complexity $O(E\log L)$ must have cost $\Omega (E\log L)$.

\vspace{2ex}
\end{abstract}




\section{Introduction}

\subsection{Background}

Two autonomous mobile entities, called agents, starting from different nodes of a network, have to meet at the same node.
This well-researched distributed task is known as {\em rendezvous}.
These mobile entities might represent human-made objects, such as software agents in computer networks or mobile robots navigating in a network of corridors in a mine. They might also be natural, such as people who want to meet in an unknown city whose streets form a network. 
The purpose of meeting might be to exchange data previously collected by the agents,
or to coordinate future network maintenance tasks, for example checking functionality of websites or of sensors forming a network. 

\subsection{Model and Problem Description}

The network is modeled as an undirected connected graph with $n$ nodes.
We seek deterministic rendezvous algorithms that do not
rely on perceiving node identifiers, and therefore can work in anonymous graphs as well  (cf. \cite{alpern02b}). 
The reason for designing such algorithms
is that, even when nodes have distinct identifiers, agents may be unable to perceive them
because of limited sensory capabilities (e.g., a mobile robot may be unable to read signs at corridor crossings), 
or nodes may be reluctant to reveal their identifiers to software agents, e.g., due to security or privacy reasons.
Note that, if nodes had distinct identifiers visible to the agents, the agents might explore the graph and meet at the node
with the smallest identifier, hence rendezvous
would reduce to graph exploration.

On the other hand, we assume that, at each node $v$,
each edge incident to $v$ has a distinct {\em port number} from 
$\{0,\dots,d-1\}$, where $d$ is the degree of $v$. These port numbers are visible to the agents.
Port numbering is {\em local} to each node, i.e., there is no relation between
port numbers at  the two endpoints of an edge. Note that in the absence of port numbers, edges incident to a node
would be undistinguishable for agents and thus rendezvous would be often impossible, 
as an adversary could prevent an agent from taking some edge incident to the current node, and this edge could be a bridge to the part of the network
where the other agent is located.
Security and privacy reasons for not revealing node identifiers to software agents are irrelevant in the case of port numbers, and 
port numbers in the case of a mine or labyrinth can be made implicit, e.g., by marking one edge at each intersection
(using a simple mark legible even by a mobile robot with very limited vision),
considering it as corresponding to port 0, and all other port numbers increasing clockwise.

Agents are initially located at different nodes of the graph and  traverse its edges in synchronous rounds.
They cannot mark visited nodes or traversed edges in any way, and they cannot communicate before meeting.
The adversary wakes up each of the agents, possibly in different rounds. 
Each agent starts executing the algorithm in the round of its wake-up.
It has a clock that ticks at each round and starts at the wake-up round of the agent.
In each round, each agent decides to either remain at the current node,
or to choose a port in order to move to one of the adjacent nodes. 
When an agent enters a node, it learns the node's degree and the port of entry. When agents cross each other
on an edge while traversing it simultaneously in different directions, they do not notice this fact.

Each agent has a distinct integer label from a fixed {\em label space} $\{1,\dots,L\}$, which it can
use in its execution of the deterministic algorithm that both agents execute. It does not know the label nor the starting round of the other agent. 
Notice that, since we study deterministic rendezvous, the absence of distinct labels precludes the possibility of meeting in highly
symmetric networks, such as rings or tori, for which there exist non-trivial port-preserving automorphisms. Indeed, in such networks,
identical agents starting simultaneously and executing the same deterministic algorithm in a distributed way will never meet, since they will be at different nodes in every round. 
In other words, assigning different labels to agents is the only way to break symmetry, as is needed to meet in every network using a deterministic algorithm. 
On the other hand, if agents knew
each other's identities, then the smaller-labelled agent could stay idle, while the other agent would try to find it. In this case rendezvous reduces to graph exploration.   
Assuming such knowledge, however, is not realistic, as agents are often created independently in different parts of the network and they know nothing about each other
prior to meeting.


The rendezvous is defined as both agents being at the same node in the same round.
Two main efficiency measures of a rendezvous algorithm are its {\em time} (the number of rounds from the start of the earlier agent until the meeting)
and its {\em cost} (the total number of edge traversals by both agents before rendezvous). 
\footnote{A different way of counting time and cost (under which our results still hold) is discussed in the Conclusion.}
We investigate tradeoffs between these measures of rendezvous performance.
A natural benchmark for both time and cost of rendezvous in a network is the time of exploration of this network by a single agent,
i.e., the worst-case number of edge traversals needed for visiting all nodes of the network,
taken over all starting nodes. Indeed, this is a lower bound on both the time and the cost of rendezvous:
an adversary can impose a large delay on one of the agents and place it at the node last explored by the other agent. 
Even for simultaneous start, there are many networks for which the best exploration time
is a lower bound on rendezvous time and cost. (One such example is oriented rings.) 
Hence we assume that some upper bound $E$ on the time of exploration starting at any node of the graph is known to the agents,
and that an agent knows how to explore the graph in time at most $E$, starting at any node of the graph. 

We express the time and cost 
of rendezvous as functions of $E$ and the size $L$ of the label space. In the Conclusion, we comment on the situation when no upper bound $E$
is known to the agents. For given parameters $E$ and $L$, we say that
a deterministic rendezvous algorithm works at a cost at most $C$ and in time at most $T$,  
if, for any two agents whose distinct labels are from the label space $\{1,\dots,L\}$ and whose initial positions are arbitrary distinct nodes in a graph that can be explored by a single agent in time $E$, the agents
meet after a total of at most $C$ edge traversals and after at most $T$ rounds since the start of the earlier agent.

A remark is in order about the value of $E$ and how it is calculated. If only an upper bound $m$ on the size of the network is known, then the
best known estimate of the time of a (log-space constructible) exploration is Reingold's \cite{Re} polynomial estimate $R(m)$ based on Universal Exploration Sequences (UXS); see also \cite{AKLLR,CDK} for solutions not log-space constructible.
The situation improves significantly if each agent has a map of the graph with unlabeled nodes, labeled ports, and the agent's starting position marked. In this case, Depth-First-Search can be performed in time at most $2n-3$,  so $E$ can be taken as $2n-3$, which is the optimal 
exploration time in networks such as the star (a tree of diameter 2). However, for some graphs a better bound $E$ can be found. For example, if the graph has a Hamiltonian cycle, then $E$ can be taken as $n-1$. If the graph has an Eulerian cycle, then $E$ can be taken as $e-1$, where $e$ is the number of edges.
Next, suppose that each agent has a
port-labeled map, but without a marked starting position.  In this case, the agent identifies on the map a DFS traversal of the graph, starting from each node
and returning to the same node.
Each DFS is a sequence of length $2n-2$ of ports (we consider the port by which each node of the traversal should be exited). From its initial position, the
agent ``tries'' each DFS one after another.
In each attempt, the agent aborts the exploration if a prescribed port is not available at the current node, and returns to the starting node. One of the attempts correctly
visits all nodes, as it is a DFS corresponding to the actual starting node of the agent, so $E$
can be taken to be $n(2n-2)$. In our study we consider $E$ to be a parameter available to both agents, together with the corresponding exploration procedure,
regardless of the particular scenario and of the sharpness of this bound.  

As far as the memory of the agent is concerned, the most demanding part of our algorithms is the underlying graph exploration. Hence, the
way in which an exploration of time at most $E$ is performed has a decisive impact on the size of the memory required. If the agent knows only an upper bound $m$ on the size of the graph and relies on a UXS to make the exploration, then exploration requires only $O(\log m)$ bits of memory
(this is the main result of \cite{Re}) but the upper bound $E$ is then fairly large, i.e., a high-degree polynomial in $m$. If the agent is given as
input a DFS walk, coded as a sequence of port numbers, starting and ending at its starting node, then the memory required to record this walk is of size 
$O(n\log n)$, but the bound $E$ is then sharper. If, given a port-labeled map of the graph with a marked starting node, the agent has to discover an efficient exploration walk by itself, then recording this map is memory-consuming, i.e., up to $O(n^2\log n)$
bits. In particular cases, e.g.,  when the underlying graph is a ring of size $n$, only $\lceil \log n \rceil$ bits of memory are needed to record $n$,
and $E$ can be made as tight as possible, i.e., $n-1$.    
However, regardless of the scenario used to organize exploration, the rest of our algorithms does not require much memory: as will be seen,  it is enough to have simple counters that can be implemented with $O(\log E +\log L)$ memory bits.

\subsection{Our Results}
\label{subsec:ourresults}

First, recall that the cost of every
rendezvous algorithm is at least $E$ and the time is at least $\Omega (E\log L)$, even for the class of rings \cite{DFKP} (for which $E=n-1$).
We present two natural rendezvous algorithms that achieve optimal cost and time, respectively, up to multiplicative constants. 
Algorithm {\tt Cheap} has cost $O(E)$
and time $O(EL)$. Algorithm {\tt Fast} has both time and cost $O(E\log L)$. 
These algorithms work for arbitrary connected graphs and arbitrary starting times of the agents.
In fact, a version of Algorithm {\tt Cheap} has cost exactly $E$ for  
 the model where the agents start simultaneously.
Our main contributions are lower bounds showing that,
perhaps surprisingly, these two algorithms achieve nearly optimal tradeoffs between the time and cost of rendezvous. 
These lower bounds  hold even in a scenario very favourable for potential rendezvous algorithms, i.e., for oriented rings of known size and with simultaneous start.
We show that any deterministic rendezvous algorithm  with time complexity $O(E\log L)$ must have cost $\Omega (E\log L)$.
Hence,  if we want to be as fast as {\tt Fast}, we cannot be cheaper.
On the other hand,
we show that
any deterministic rendezvous algorithm of cost asymptotically $E$ (i.e., of cost $E+o(E)$) must have time $\Omega(EL)$.
Hence, in the model with simultaneous start, if we want to be as cheap as {\tt Cheap}, we cannot be faster. 

It is natural to ask if it is possible to solve rendezvous both at cost $o(E\log L)$, i.e., beating the cost of Algorithm {\tt Fast},
and in time $o(EL)$, i.e., beating the time of Algorithm {\tt Cheap}. It turns out that the answer to this question is ``yes''. 
Indeed, we provide an algorithm called {\tt FastWithRelabeling} that works at cost $O(E)$ and in time $o(EL)$. Moreover,  this shows a separation between the time necessary to solve 
rendezvous  at cost asymptotically $E$, i.e., at cost $E+o(E)$, and the time sufficient to solve 
rendezvous  at cost $\Theta(E)$. In the first case, the lower bound $\Omega(EL)$ on time holds, while in the second it does not.


\subsection{Related Work}
\label{subsec:relatwork}

Exploration and rendezvous are the two main tasks accomplished by mobile agents in networks modeled as graphs.
Algorithms for graph exploration by mobile agents (often called robots) have been
intensely studied in recent literature. A lot of  research is concerned with the case of a
single agent exploring a labeled graph.  In \cite{AH,BFRSV,BS,DP,FT} the
agent explores strongly-connected directed graphs. In a directed graph, an agent can move only
in the direction from tail to head of a directed edge, not vice-versa.  In
particular, \cite{DP} investigates the minimum time of exploration of
directed graphs, and \cite{AH,FT} give improved algorithms for this
problem in terms of the deficiency of the graph (i.e., the minimum
number of directed edges to be added to make the graph Eulerian).  Many papers,
e.g., \cite{GPRZ,DePe,DKK,PaPe} study the scenario where the
explored graph is labeled and undirected, and the agent can traverse edges in both
directions.  In
\cite{PaPe}, it is shown that a graph with $n$ nodes and $e$ edges can
be explored in time $e+O(n)$.  In some papers, additional restrictions
on the moves of the agent are imposed.  It is assumed that the agent
has either a restricted tank \cite{ABRS,BRS2}, forcing it to
periodically return to the base for refueling, or that it is tethered,
i.e., attached to the base by a rope or cable of restricted length
\cite{DKK}.
In \cite{DePe}, the authors investigate the problem of how
the availability of a map influences the efficiency of exploration.
In \cite{AKLLR}, the authors proved the existence of a polynomial-time deterministic exploration for all graphs with a given bound on size.
In \cite{Re}, a log-space construction of such an exploration was shown.
 
In all the above papers, except \cite{BS}, 
exploration is performed by
a single agent.  Deterministic exploration by many agents has been investigated
mostly in the context when the moves of the agents are centrally
coordinated.  In \cite{FHK}, approximation algorithms are given for
the collective exploration problem in arbitrary graphs. In
\cite{AB1,AB2}, the authors construct approximation algorithms for the
collective exploration problem in weighted trees. On the other hand,
in \cite{FGKP}, the authors study the problem of distributed collective
exploration of trees of unknown topology. In \cite{DieuPe}, exploration of arbitrary networks by many anonymous agents is investigated,
while in \cite{DDKPU}, this task is studied for labeled agents and labeled nodes.

The problem of rendezvous has been studied both under randomized and deterministic scenarios.
An extensive survey of  randomized rendezvous in various models  can be found in
\cite{alpern02b}, cf. also  \cite{alpern95a,alpern02a,anderson90,baston98,israeli}. 
Deterministic rendezvous in networks has been surveyed in \cite{Pe}.
Several authors
considered geometric scenarios (rendezvous in an interval of the real line, e.g.,  \cite{baston98,baston01},
or in the plane, e.g., \cite{anderson98a,anderson98b}).
Gathering more than two agents was studied, e.g., in \cite{fpsw,israeli,lim96,thomas92}.

For the deterministic setting many authors studied the feasibility and time complexity of rendezvous. For instance, deterministic rendezvous of agents equipped with tokens used to mark nodes was considered, e.g., in~\cite{KKSS}. 
Deterministic rendezvous in rings by labeled agents, without the ability to mark nodes, was investigated, e.g., in \cite{DFKP,KM}. In \cite{DFKP}, the authors gave tight upper and lower bounds
of $\Theta (D\log \ell)$ on the 
time of rendezvous when agents start simultaneously, where $D$ is the initial distance between agents and $\ell$ is the smaller label. 
They also gave a lower bound of $\Omega(n+D\log \ell)$ on the time of rendezvous with arbitrary delay between the agents' starting times in $n$-node rings.  In \cite{KM} an upper bound $O(n\log \ell)$
on the time of rendezvous was given, even without knowledge of $n$.
Most relevant to our work are the results about 
deterministic rendezvous in arbitrary graphs, when the two agents cannot mark nodes, but have unique labels  \cite{DFKP,KM,TSZ07}.
In \cite{DFKP}, the authors present a rendezvous algorithm whose running time is polynomial in the size of the graph, in the length of the shorter
label and in the delay between the starting times of the agents. In \cite{KM,TSZ07}, rendezvous time is polynomial in the first two of these parameters and independent of the delay between the starting times.

Memory required by the agents to achieve deterministic rendezvous was studied in \cite{FP2} for trees and in  \cite{CKP} for general graphs.
Memory needed for randomized rendezvous in the ring is discussed, e.g., in~\cite{KKPM08}. 

Apart from the synchronous model used in this paper, several authors investigated asynchronous rendezvous in the plane \cite{CFPS,fpsw} and in network environments
\cite{BCGIL,CLP,DGKKP,DPV}.
In the latter scenario, the agent chooses the edge to traverse, but the adversary controls the speed of the agent. Under this assumption, rendezvous
at a node cannot be guaranteed even in very simple graphs. Hence the rendezvous requirement is relaxed to permit the agents to meet inside an edge.

\section{Algorithms}

In this section we present three rendezvous algorithms: Algorithm {\tt Cheap}, Algorithm {\tt Fast}, and Algorithm {\tt FastWithRelabeling}$(s)$ for any function $s(L)\leq L$.
In each case, we first describe the algorithm in the easier case of simultaneous start, give a general formulation for arbitrary starting times of the agents, prove its correctness, and establish its time and cost complexities.

Assume that each agent $X$ is given a distinct label $\ell_X$ from the set $\{1,\ldots,L\}$. Let \texttt{EXPLORE} be a procedure that, for every possible starting node, takes $E$ rounds to perform an exploration of the entire input graph. If the exploration is completed earlier, the agent waits after finishing it until a total of $E$ rounds have elapsed. Upon meeting, both agents stop.

We start with the description of a version of Algorithm {\tt Cheap} for the model where the agents start simultaneously. Agent $X$ waits $(\ell _X-1)E$ rounds and then explores the graph once.

To see why this works, assume, without loss of generality, that $\ell_A < \ell_B$. Then, agent $B$ waits at its starting node in rounds $\{1,\ldots,(\ell_B-1)E\} \supseteq \{1,\ldots,\ell_AE\}$, and agent $A$ explores the entire graph in rounds $\{(\ell_A-1)E+1,\ldots,\ell_AE\}$. Therefore, agent $A$ meets agent $B$ at its starting node by round $\ell_AE$. Thus,  rendezvous is achieved in at most $\ell E$ rounds, where $\ell$ is the smaller label.
In the worst case this is $(L-1)E$. Since at most one exploration is performed, the cost is at most $E$.

In the general case of arbitrary starting times of the agents, Algorithm {\tt Cheap} is described as follows.

\begin{algorithm}[H]
\caption{\texttt{Cheap($\ell$,\texttt{EXPLORE})}}
\begin{algorithmic}[1]
\STATE \textrm{Execute \texttt{EXPLORE} once}
\STATE \textrm{Wait $2\ell E$ rounds}
\STATE \textrm{Execute \texttt{EXPLORE} once}
\end{algorithmic}
\end{algorithm}

\begin{proposition}\label{cheap}
Algorithm {\tt Cheap} completes rendezvous with cost at most $3E$ and in time at most $(2L+1)E$.
\end{proposition}

\begin{proof}
Suppose that agent $A$ starts its execution in round 1 and that agent $B$ starts its execution in round $\tau$ for some $\tau \geq 1$. From the algorithm's specification, we can deduce the following:
\begin{itemize}
\item Agent $A$'s first exploration (i.e., Line 1) starts in round $1$ and ends in round $E$, its waiting period (i.e., Line 2) starts in round $E+1$ and ends in round $(2\ell_A+1)E$, and its second exploration (i.e., Line 3) starts in round $(2\ell_A+1)E+1$ and ends in round $(2\ell_A+2)E$.
\item  Agent $B$'s first exploration starts in round $\tau$ and ends in round $\tau+E-1$, its waiting period starts in round $\tau+E$ and ends in round $\tau+(2\ell_B+1)E-1$, and its second exploration starts in round $\tau+(2\ell_B+1)E$ and ends in round $\tau+(2\ell_B+2)E-1$.
\end{itemize}
 
First, observe that, if $B$'s start is significantly delayed, then agent $A$ meets agent $B$ during agent $A$'s first exploration of the graph. Namely, if $\tau > E$, then the agents meet within the first $E$ rounds.

So, in what follows, we assume that $\tau \leq E$. Since $1 \leq \tau \leq E$, $B$'s second exploration occurs completely within the time segment $[(2\ell_B+1)E+1,\ldots,(2\ell_B+3)E-1]$.

If $\ell_A > \ell_B$, then $A$'s waiting period ends in round $(2\ell_A+1)E \geq (2(\ell_B+1)+1)E = (2\ell_B+3)E$. Also, note that $A$'s waiting period starts in round $E+1 \leq \ell_BE+1$. Therefore, agent $A$ is idle throughout the time segment $[\ell_BE+1,\ldots,(2\ell_B+3)E] \supseteq [(2\ell_B+1)E+1,\ldots,(2\ell_B+3)E-1]$. Hence, agent $B$ meets agent $A$ by round $(2\ell_B+3)E-1$.

If $\ell_B > \ell_A$, then $B$'s waiting period ends in round $\tau + (2\ell_B+1)E - 1 \geq \tau + (2\ell_A+3)E - 1$. Also, $B$'s waiting period starts in round $\tau + E \leq \tau + \ell_AE$. Since $1 \leq \tau \leq E$, $B$ is idle throughout the time segment $[\tau + \ell_AE,\ldots,\tau+(2\ell_A+3)E-1] \supseteq [(\ell_A+1)E,\ldots,(2\ell_A+3)E]$. However, $A$'s second exploration occurs during the time segment 
$[(2\ell_A+1)E+1,\ldots,(2\ell_A+2)E)] \subseteq [(\ell_A+1)E,\ldots,(2\ell_A+3)E]$. Hence, agent $A$ meets agent $B$ by round $(2\ell_A+2)E$.

Thus, Algorithm {\tt Cheap} completes rendezvous using at most $(2\ell+3)E$ rounds, where $\ell$ is the smaller label. In the worst case, this is $(2L+1)E$.
Since the meeting occurs before the start of the second exploration of the agent with the larger label, the total cost of the algorithm is at most $3E$.
\end{proof}

Next, in order to describe Algorithm {\tt Fast}, we recall the label transformation from \cite{DPV}. If $x=(c_1\cdots c_r)$ is the binary representation of the label $\ell$
of an agent, define the {\em modified label} of the agent to be the sequence $M(\ell)=(c_1c_1c_2c_2\cdots c_rc_r01)$.  
Note that, for any distinct $x$ and $y$, the sequence $M(x)$ is never a prefix of $M(y)$.
Also, $M(x) \neq M(y)$ if $x\neq y$. Since the (original) labels of the agents are different, there exists an index for which their transformed labels differ.
Note that if $z=1+\lfloor \log \ell \rfloor$ is the length of the binary representation of the label $\ell$ of the agent, then $m=2z +2$ is the length of its modified label.

We describe Algorithm {\tt Fast}, first in the case of simultaneous start. Suppose that $(b_1\cdots b_m)$ is the transformed label of an agent.   
In the time segment $[(i-1)E+1,iE]$, the agent executes $\texttt{EXPLORE}$ if $b_i=1$, and, otherwise, the agent stays idle.


To see why this works, consider any two agents $A$ and $B$, and let $S_A$ and $S_B$ denote their transformed labels, respectively. Consider the smallest index $j$ such that $S_A[j] \neq S_B[j]$. Without loss of generality, assume that $S_A[j] = 1$ and $S_B[j] = 0$. It follows that, during the time segment $[(j-1)E+1,\ldots,jE]$, agent $A$ explores the entire graph while $B$ is idle. Therefore, agent $A$ meets agent $B$ by round $jE$. Hence the worst possible time is $(2 \lfloor \log (L-1)\rfloor+4)E=O(E\log L)$.
The cost is bounded above by twice the time, hence it is also $O(E\log L)$. 

In the general case of arbitrary starting times Algorithm {\tt Fast} is described as follows.

\begin{algorithm}[H]
\caption{\texttt{Fast($\ell$,\texttt{EXPLORE})}}
\begin{algorithmic}[1]
\STATE $S[1 \ldots m] \leftarrow M(\ell)$
\STATE $T[1 \ldots 2m+1] \leftarrow (1,S[1],S[1],S[2],S[2],\ldots,S[m],S[m])$ 
\FOR{$i = 1$ to $2m+1$}
	\IF{$(T[i] = 1)$}
		\STATE{\textrm{execute \texttt{EXPLORE} once}}
	\ELSE
		\STATE{\textrm{wait $E$ rounds}}
	\ENDIF
\ENDFOR
\end{algorithmic}
\end{algorithm}

\begin{proposition}\label{fast}
Algorithm {\tt Fast} completes rendezvous with cost at most $(8\log{(L-1)}+18)E$ and in time at most $(4\log{(L-1)}+9)E$.
\end{proposition}
\begin{proof}
As the cost is bounded above by twice the time, it is sufficient to analyze time.
Consider any two agents $A$ and $B$. For each agent $X$, let $S_X=M(\ell_X)$, and let $m$ be the length of $S_X$. Let $T_X$ be the string of length $2m+1$ such that $T_X[1]=1$, and, for each $i \in \{2,\ldots,m\}$, $T_X[2i] = T_X[2i+1] = S_X[i]$.

Suppose that agent $A$ starts its execution in round 1 and that agent $B$ starts its execution in round $\tau$ for some $\tau \geq 1$. First, observe that, if $B$'s start is significantly delayed, then agent $A$ meets agent $B$ during agent $A$'s first exploration of the graph. Namely, if $\tau > E$, then the agents meet within the first $E$ rounds. So, in what follows, we assume that $\tau \leq E$. Consider the smallest $j$ such that $S_A[j] \neq S_B[j]$.

First, suppose that $S_A[j] = 0$. It follows that $T_A[2j] = T_A[2j+1] = 0$, so $A$ is idle during the time segment $[(2j-1)E+1,\ldots,(2j+1)E]$. Also, $T_B[2j] =1$, so $B$ performs procedure {\tt EXPLORE} starting in round $(2j-1)E+\tau+1$ and ending in round $2jE+\tau$. Since $0 \leq \tau \leq E$, this execution of {\tt EXPLORE} is completely contained in the time segment $[(2j-1)E+1,\ldots,(2j+1)E]$. Therefore, $B$ meets $A$ by round $(2j+1)E$. 

Next, suppose that $S_A[j] = 1$. It follows that $T_B[2j] = T_B[2j+1] = 0$, so $B$ is idle during the time segment $[(2j-1)E + \tau +1,\ldots,(2j+1)E + \tau]$. Since $0 \leq \tau \leq E$, this interval contains the time segment $[2jE+1,\ldots,(2j+1)E]$. Also, $T_A[2j+1] = 1$, so $A$ performs  procedure {\tt EXPLORE} starting in round $2jE+1$ and ending in round $(2j+1)E$. Therefore, $A$ meets $B$ by round $(2j+1)E$.

Hence, the two agents meet by round $(2j+1)E$, and thus, the worst possible time is $(4 \lfloor \log (L-1)\rfloor+9)E \in O(E\log L)$.
\end{proof}

The worst-case cost of Algorithm {\tt Fast} occurs when the binary representation of an agent's label has large weight, i.e., has many 1's. We can reduce the cost if we relabel the agents in such a way that all labels have small weight. This motivates the following algorithm called {\tt FastWithRelabeling}.

For any function $w:\mathbb{N} \longrightarrow \mathbb{N}$ such that $w(L) \leq L$, we define Algorithm {\tt FastWithRelabeling}$(w)$ as follows. Let $t$ be the smallest positive integer such that $\binom{t}{w(L)} \geq L$. For any set 
$A \subset \{1,\ldots,t\}$, the characteristic function $\chi _A: \{1,\ldots,t\} \longrightarrow \{0,1\}$ is defined by $\chi _A (i)=1$ if and only if $i \in A$. Each characteristic function $\chi_A$ yields a $t$-bit binary string $s_A$ where the $i$'th bit of $s_A$ is equal to $\chi_A(i)$.
We say that a set $A \subset \{1,\ldots,t\}$ is lexicographically smaller than a set $B \subset \{1,\ldots,t\}$ if $s_A$ is lexicographically smaller than $s_B$. Each agent $X$ is assigned the lexicographically $\ell_X$-th smallest $w(L)$-subset of $\{1,\ldots,t\}$, and its new label $\ell_X'$ is taken to be the $t$-bit binary string corresponding to the characteristic function of this set. Then, Algorithm {\tt Fast} is executed with the new labels.

\begin{proposition}\label{fastwithrelabeling}
Algorithm {\tt FastWithRelabeling}$(w)$ completes rendezvous with cost at most $(2\cdot w(L))E$ and in time at most $(4t+ 5)E$, where $t$ is the smallest positive integer such that $\binom{t}{w(L)} \geq L$.
\end{proposition}
\begin{proof}
We note that, for two distinct agents $A$ and $B$, we have $\ell_A' \neq \ell_B'$. This is because $\ell_A \neq \ell_B$, and, by the choice of $t$, there are at least $L$ subsets of $\{1,\ldots,t\}$ of size $w(L)$, so $A$ and $B$ are assigned distinct subsets of $\{1,\ldots,t\}$. Using the same proof of correctness and worst-case time analysis as Algorithm {\tt Fast}, with labels of fixed length $t$ instead of length at most $1+\log{(L-1)}$, it follows that Algorithm {\tt FastWithRelabeling} correctly solves rendezvous in time at most $(4t+5)E$. To analyze the cost, we note that each label has exactly $w(L)$ 1's, so the combined cost incurred by the two agents is at most $(2\cdot w(L))E$.
\end{proof}

The following corollary shows that Algorithm {\tt FastWithRelabeling}$(w)$, for constant functions $w(L)=c$ where $c>1$, solves rendezvous
at cost $O(E)$ and in time $o(EL)$.

\begin{corollary}\label{constantweight}
For any positive integer function $w \in O(1)$, Algorithm {\tt FastWithRelabeling}$(w)$ works with cost $O(E)$ and in time $O(L^{1/w(L)}E)$.
\end{corollary}
\begin{proof}
Let $w(L) = c$ for some positive constant integer $c$. Let $t' = c \cdot L^{1/c}$. Then $\binom{t'}{w(L)} = \binom{c \cdot L^{1/c}}{c} \geq \left(\frac{c\cdot L^{1/c}}{c}\right)^c = L$. Therefore, $t \leq t' = c \cdot L^{1/c}$. By Proposition \ref{fastwithrelabeling}, the worst-case time of Algorithm {\tt FastWithRelabeling}$(w)$ is at most $(4c\cdot L^{1/c} + 5)E \in O(L^{1/w(L)}E)$, and the worst-case cost is at most $2cE \in O(E)$.
\end{proof}

\section{Lower Bounds}

In order to make our lower bounds as strong as possible, we show that they hold even in a very restricted situation: when the underlying graph is particularly simple and the agents have full knowledge of it. A ring is {\em oriented} if every edge has port labels 0 and 1 at the two end-points.
Such a port labeling induces orientation of the ring: at each node, we will say that taking port 0 is going clockwise and taking port 1 is going counterclockwise.
Throughout this section, we assume that agents operate in an oriented ring of size $n$ known to the agents. Hence, in this case, $E$
is taken as $n-1$: starting from any node an agent can explore the ring going $n-1$ steps clockwise.
This is, of course, an optimal exploration.  Moreover, we assume that both agents start simultaneously, i.e., their clock values are equal in each round.
Even in this scenario, which is very favourable to potential rendezvous algorithms, we establish lower bounds proving that
our algorithms {\tt Cheap} and {\tt Fast} capture the time vs. cost tradeoffs for rendezvous almost tightly. 

In our lower bound proofs, we use the following terminology. For simplicity, an agent with label $x$ will be called agent $x$. 
Consider a rendezvous algorithm $\mathcal{A}$.
Consider two arbitrary agents $x,y$ and two arbitrary nodes $p_x,p_y$ in the oriented ring of size $n$. We denote by $\alpha(x,p_x,y,p_y)$ the execution of algorithm
$\mathcal{A}$  in which $x$ starts at node $p_x$ and $y$ starts at node $p_y$. The final round of $\alpha(x,p_x,y,p_y)$, denoted by $|\alpha(x,p_x,y,p_y)|$, is the first round in which $x$ and $y$ meet. In a slight abuse of notation, we denote by $\alpha(x,p_x,\bot,\bot)$ the {\em solo} execution of $\mathcal{A}$, i.e., when $x$ executes the algorithm alone, starting at node $p_x$. Note that the behaviour of agent $x$ in an execution $\alpha(x,p_x,y,p_y)$ is the same
as its behaviour in execution $\alpha(x,p_x,\bot,\bot)$ until round $|\alpha(x,p_x,y,p_y)|$.

For each label $x \in \{1,\ldots,L\}$, algorithm $\mathcal{A}$ specifies a \emph{behaviour vector} $V_x$. In particular, $V_x$ is a sequence with terms from $\{-1,0,1\}$ that specifies, for each round $i$ of the solo execution of agent $x$, whether agent  $x$ moves clockwise (denoted by $1$), remains idle (denoted by $0$), or moves counter-clockwise (denoted by $-1$). Note that an agent's behaviour vector is independent of its starting position, since an agent cannot determine where on the ring it is initially positioned.  

We now describe a procedure {\tt Trim}($\mathcal{A}$) which modifies the behaviour vectors specified by $\mathcal{A}$. At a high level, we are zeroing the entries that the algorithm never uses so that, if we show the existence of a non-zero entry in round number $i$ of some behaviour vector, then there is an execution of the algorithm that takes at least $i$ rounds. Specifically, for each $x \in \{1,\ldots,L\}$:
\begin{enumerate}
\item Find the maximum value of $|\alpha(x,p_x,y,p_y)|$, taken over all $y \in \{1,\ldots,L\} \setminus \{x\}$ and nodes $p_x,p_y$. Denote this maximum by $m_x$.
\item For all $j>m_x$, set $V_x[j]=0$.
\end{enumerate}
Note that this does not change any non-solo execution of $\mathcal{A}$: any modified entry in $V_x$ corresponds to a round that occurs after $x$ has met with any other agent. Also, after performing this trimming operation, for any non-zero entry $V_x[i]$, there exists an agent $y$ and there exist starting positions for $x$ and $y$ such that $x$ and $y$ have not met by round $i$ and agent $x$ moves during round $i$. We obtain lower bounds on the running time (or cost) of $\mathcal{A}$ by proving lower bounds on the length (or weight) of behaviour vectors resulting from procedure  {\tt Trim}($\mathcal{A}$).

Our first lower bound shows that no rendezvous algorithm of cost asymptotically $E$ (i.e., of cost $E +o(E)$), can beat the time $\Theta(EL)$
of Algorithm {\tt Cheap}. (Recall that Algorithm {\tt Cheap} always has cost $O(E)$ and it has cost exactly $E$ in a model with simultaneous start.)

\begin{theorem}
Any deterministic rendezvous algorithm of cost $E +o(E)$ must have time $\Omega(EL)$.
\end{theorem}

\begin{proof}
Let $\mathcal{A}$ be a rendezvous algorithm such that, for some $\varphi \in o(E)$, for every pair of agent labels, and for every pair of starting positions of the agents, rendezvous is completed at cost at most $E+\varphi$. As previously explained, instead of behaviour vectors of algorithm $\mathcal{A}$, we consider behaviour vectors
resulting from procedure {\tt Trim}($\mathcal{A}$).

For any execution $\alpha$, let $seg(x,\alpha)$ be the segment of the ring that agent $x$ explores during execution $\alpha$, and denote by $|seg(x,\alpha)|$ the number of edges in this segment.

During any particular round of an execution $\alpha$, we can determine on which `side' of its starting position the agent is currently situated. More specifically, in any round $i$ of $\alpha$, if the prefix of an agent's behaviour vector up to round $i$ has at least as many (resp. at most as many) $-1$'s as $1$'s, then we say that the agent is on its \emph{counterclockwise side} (resp. {\em clockwise side}) in round $i$. Let $seg_{-1}(x,\alpha)$ be the segment of the ring that agent $x$ explores while on its counterclockwise side during execution $\alpha$, and denote by $|seg_{-1}(x,\alpha)|$ the number of edges in this segment. Similarly, let $seg_1(x,\alpha)$ be the segment of the ring that agent $x$ explores while on its clockwise side during execution $\alpha$, and denote by $|seg_1(x,\alpha)|$ the number of edges in this segment. Note that $seg(x,\alpha) = seg_1(x,\alpha) \cup seg_{-1}(x,\alpha)$, hence we have $|seg(x,\alpha)| \leq |seg_1(x,\alpha)|+|seg_{-1}(x,\alpha)|$. 

Note that 
$|seg_{-1}(x,\alpha(x,p_x,\bot,\bot))|$ and $|seg_1(x,\alpha(x,p_x,\bot,\bot))|$
do not depend on the choice of $p_x$, since, in a solo execution, the agent's behaviour is the same regardless of its starting node. If  $|seg_{-1}(x,\alpha(x,p_x,\bot,\bot))| \geq |seg_1(x,\alpha(x,p_x,\bot,\bot))|$, we say that agent $x$ is \emph{counter-clockwise-heavy}. Otherwise, we say that agent $x$ is \emph{clockwise-heavy}. Without loss of generality, we assume that at least half of the agents are clockwise-heavy, and we proceed by considering only the clockwise-heavy agents.

For any agent $x$ and for any node $p_x$, let $forward(x)$ be the number of edges in $seg_{1}(x,\alpha(x,p_x,\bot,\bot))$ and let $back(x)$ be the number of edges in $seg_{-1}(x,\alpha(x,p_x,\bot,\bot))$. Since we consider only clockwise-heavy agents, we have $back(x) \leq forward(x)$.
For any agent $x$ and any execution $\alpha$, let $cost(x,\alpha)$ be the number of edge traversals performed by $x$ during execution $\alpha$. 

\begin{fact}\label{MakeDisjoint}
Consider two agents $A,B$ and two nodes $p_A,p_B$ such that $|seg(A,\alpha(A,p_A,B,p_B))|+|seg(B,\alpha(A,p_A,B,p_B))| < E$. Then, for some node $p_B'$, during the first $|\alpha(A,p_A,B,p_B)|$ rounds of  $\alpha(A,p_A,B,p_B')$, the segments $seg(A,\alpha(A,p_A,B,p_B'))$ and $seg(B,\alpha(A,p_A,B,p_B'))$ are disjoint.
\end{fact}
If the nodes are labeled $0,\ldots,n-1$ in the clockwise direction, then choosing $p_B' = p_A + forward(A) + 1 + back(B) (\!\!\!\!\mod n)$ verifies the above fact.

\begin{fact}\label{fact:SoloCost}
For any agent $A$ and any node $p_A$, $cost(A,\alpha(A,p_A,\bot,\bot)) \geq 2back(A) + forward(A)$.
\end{fact}
To see why, note that, in a solo execution, agent $A$ must visit all edges in $seg(A,\alpha(A,p_A,\bot,\bot))$. To do so, there must be a round in which $A$ returns to $p_A$ after reaching one of the endpoints of $seg(\alpha(A,p_A,\bot,\bot))$. Therefore, $A$ must visit all of the edges in $seg_{-1}(\alpha(A,p_A,\bot,\bot))=back(A)$ at least twice, or all of the edges in $seg_{1}(\alpha(A,p_A,\bot,\bot))=forward(A)$ at least twice. By assumption, $back(A) \leq  forward(A)$, which implies the fact.

\begin{fact}\label{fact:SmallMin}
For any agent $A$, $back(A) \leq \varphi$.
\end{fact}
We prove this fact by contradiction. Assume that, for some agent $A$, $back(A) > \varphi$. Recall, from the trimming of algorithm $\mathcal{A}$, that $m_A$ is defined to be the maximum value of $|\alpha(A,p_x,y,p_y)|$, taken over all $y \in \{1,\ldots,L\} \setminus \{A\}$ and nodes $p_x,p_y$. Choose $p_A,B,p_B$ such that $|\alpha(A,p_A,B,p_B)| = m_A$. Let $\alpha = \alpha(A,p_A,B,p_B)$, and let $\alpha_A = \alpha(A,p_A,\bot,\bot)$.

In the trimmed version of $\mathcal{A}$, $V_A[i]=0$ for all $i > m_A$. Therefore, $A$'s behaviour is identical in both $\alpha$ and $\alpha_A$. In particular, this implies that $cost(A,\alpha) = cost(A,\alpha_A)$ and $seg(A,\alpha) = seg(A,\alpha_A)$. By Fact \ref{fact:SoloCost}, it follows that $cost(A,\alpha) = 2back(A) + forward(A) + \delta$ for some $\delta \geq 0$. Also, since $|seg(A,\alpha_A)| \leq |seg_1(A,\alpha_A)|+|seg_{-1}(A,\alpha_A)|$, it follows that $|seg(A,\alpha)| \leq back(A)+forward(A)$.

Next, note that $|seg(B,\alpha)| \leq cost(B,\alpha)$. Further, since the combined costs incurred by $A$ and $B$ in execution $\alpha$ are at most $E+\varphi$, we get that $cost(B,\alpha) \leq E+\varphi-cost(A,\alpha)$. Thus, $|seg(B,\alpha)| \leq E+\varphi-2back(A)-forward(A)-\delta$. It follows that $|seg(A,\alpha)| + |seg(B,\alpha)| \leq E+\varphi-back(A)-\delta$. By assumption, $\varphi-back(A) < 0$, so we get that $|seg(A,\alpha)| + |seg(B,\alpha)| < E$. 

By Fact \ref{MakeDisjoint}, there is a node $p_B'$ such that, if we execute $\alpha(A,p_A,B,p_B')$ for $m_A$ rounds, then the set of edges traversed by $A$ and the set of edges traversed by $B$ are disjoint. It follows that $A$ and $B$ do not meet during the first $m_A$ rounds of execution $\alpha(A,p_A,B,p_B')$. By the definition of $m_A$, there is no choice of $p_x,y,p_y$ such that $|\alpha(A,p_x,y,p_y)| > m_A$. Therefore, $A$ and $B$ do not meet in execution $\alpha(A,p_A,B,p_B')$, which contradicts the correctness of $\mathcal{A}$. This completes the proof of Fact \ref{fact:SmallMin}.




Starting with an arbitrary node, label the nodes of the ring using the integers $0,\ldots,n-1$, ascending in the clockwise direction.
This is for analysis only: the agents do not have access to any node labeling.
For any execution $\alpha$ involving an agent $A$, let $disp(A,\alpha) = \sum_{j=1}^{|\alpha|} V_A[j]$. In other words, $disp(A,\alpha)$ is the displacement of agent $A$ in the clockwise direction at the end of execution $\alpha$. The following fact follows from this definition.

\begin{fact}\label{fact:BoundDisp}
For any execution $\alpha$ involving an agent $A$, $-back(A) \leq disp(A,\alpha) \leq forward(A)$.
\end{fact}

Let $F = \lceil E/2 \rceil$. 
For any execution $\alpha$ involving agents $A$ and $B$, we say that an agent $A$ is \emph{eager}  if $disp(A,\alpha) \geq disp(B,\alpha)+F$.

\begin{fact}\label{eager}
Consider any two agents $A,B$. In the execution $\alpha(A,0,B,F)$, exactly one of $A$ or $B$ is eager.
\end{fact}
To see why, first note that it cannot be the case that both $A$ and $B$ are eager. Next, if neither agent is eager, then, at the end of the execution, the number of edges that separate $A$ and $B$ is at least $F - |disp(A,\alpha)-disp(B,\alpha)| > 0$, which contradicts rendezvous. This completes the proof of the fact.

A directed graph $G$ is a \emph{tournament} if, for each pair of distinct vertices $a,b \in V(G)$, exactly one of $(a,b)$ or $(b,a)$ is an edge in $E(G)$.
We construct a tournament graph $T$ with $\lfloor \frac{L}{2} \rfloor$ vertices, as follows. First, assign to each vertex in $T$ a unique label from the set of clockwise-heavy agents. Next, for each pair of vertices $A,B$ in $T$, with $A<B$, we add a directed edge between $A$ and $B$ whose tail is the eager agent in $\alpha(A,0,B,F)$.
By Fact \ref{eager}, this operation is well-defined.
Every tournament graph has a directed Hamiltonian path \cite{Red}. Let $(A_1,\ldots,A_{\lfloor \frac{L}{2} \rfloor})$ be the sequence of agent labels encountered along one such path. For each $i \in \{1,\ldots,\lfloor \frac{L}{2} \rfloor-1\}$, let $\alpha_i = \alpha(\min\{A_i,A_{i+1}\},0,\max\{A_i,A_{i+1}\},F)$. 
This definition ensures that $\alpha_i$ is the execution that was used to define the directed edge $(A_i,A_{i+1})$ in the tournament graph. 

\begin{fact}\label{BoundedDisp}
For each $i \in \{1,\ldots,\lfloor \frac{L}{2} \rfloor-1\}$, $disp(A_{i+1},\alpha_i) \leq (F+\varphi)/2$.
\end{fact}
In order to prove this fact, note that $A_i$ is the eager agent in execution $\alpha_i$. Therefore, $disp(A_i,\alpha_i) \geq disp(A_{i+1},\alpha_i) + F$. It follows that the total cost incurred by the two agents in execution $\alpha_i$ is at least $disp(A_i,\alpha_i) + disp(A_{i+1},\alpha_i) \geq 2disp(A_{i+1},\alpha_i) + F$. Thus, $2disp(A_{i+1},\alpha_i)+F \leq E+\varphi \leq 2F+\varphi$, so $disp(A_{i+1},\alpha_i) \leq (F+\varphi)/2$. This completes the proof of the fact.

\begin{fact}\label{longer}
For each $i \in \{1,\ldots,\lfloor \frac{L}{2} \rfloor-1\}$, $|\alpha_{i+1}| > |\alpha _i|$. 
\end{fact}

In order to prove this fact, assume, for the purpose of contradiction, that we have $|\alpha_{i+1}| \leq |\alpha _i|$. Since $A_{i+1}$ is eager in execution $\alpha_{i+1}$, we have $disp(A_{i+1},\alpha_{i+1}) \geq disp(A_{i+2},\alpha_{i+1})+F$, and, by Facts \ref{fact:SmallMin} and \ref{fact:BoundDisp}, it follows that $disp(A_{i+1},\alpha_{i+1}) \geq F-\varphi$. By the assumption that $|\alpha_{i+1}| \leq |\alpha _i|$, it follows that at time $|\alpha_{i+1}|$ in execution $\alpha _i$, agent $A_{i+1}$ has a positive (clockwise) displacement of at least $F-\varphi$, and it incurred a cost of at least $F-\varphi$. Since $A_i$ is eager in execution $\alpha_i$, and the initial distance between the two agents is $F$, it follows that the two agents incur an additional cost of $(F-\varphi)+F$ during execution $\alpha_i$ in order for rendezvous to occur. Hence, the total cost incurred by both agents in execution $\alpha_i$ is at least $2(F-\varphi)+F = 3F-2\varphi > 2F+\varphi \geq E+\varphi$, a contradiction. 


\begin{fact}\label{ind}
For each $i \in \{1,\ldots,\lfloor \frac{L}{2} \rfloor-1\}$, $|\alpha_i| \geq i\left(\frac{F-3\varphi}{2}\right)$. 
\end{fact}
We prove this fact by induction on $i$. 
For the base case, note that, in execution $\alpha _1$, the time needed for rendezvous is at least $F/2$, hence $|\alpha_1| \geq F/2 \geq \frac{F-3\varphi}{2}$.

Next, as induction hypothesis, assume that for some $i \in \{1,\ldots,\lfloor \frac{L}{2} \rfloor-2\}$, $|\alpha_i| \geq i\left(\frac{F-3\varphi}{2}\right)$. Consider the execution $\alpha_{i+1}$. From Fact \ref{BoundedDisp}, $disp(A_{i+1},\alpha_i) = \sum_{j=1}^{|\alpha_i|} V_{A_{i+1}}[j] \leq (F+\varphi)/2$. However, in execution $\alpha_{i+1}$, agent $A_{i+1}$ is eager, so $\sum_{j=1}^{|\alpha_{i+1}|} V_{A_{i+1}}[j] = disp(A_{i+1},\alpha_{i+1}) \geq disp(A_{i+2},\alpha_{i+1})+F$. By Facts \ref{fact:SmallMin} and \ref{fact:BoundDisp},  $disp(A_{i+2},\alpha_{i+1})+F \geq F-\varphi$. So, we have shown that $F-\varphi \leq \sum_{j=1}^{|\alpha_{i+1}|} V_{A_{i+1}}[j] = \left[ \sum_{j=1}^{|\alpha_{i}|} V_{A_{i+1}}[j] \right] + \left[ \sum_{j=|\alpha_i|+1}^{|\alpha_{i+1}|} V_{A_{i+1}}[j] \right] \leq \left[(F+\varphi)/2\right] + \left[ \sum_{j=|\alpha_i|+1}^{|\alpha_{i+1}|} V_{A_{i+1}}[j] \right]$. (Note that the above decomposition of the sum into two sub-sums is possible in view of Fact \ref{longer}). It follows that $\sum_{j=|\alpha_i|+1}^{|\alpha_{i+1}|} V_{A_{i+1}}[j] \geq \frac{F-3\varphi}{2}$, so $|\alpha_{i+1}|-|\alpha_i| \geq \frac{F-3\varphi}{2}$. Finally, by the induction hypothesis, we get that $|\alpha_{i+1}| = \left(|\alpha_{i+1}|-|\alpha_i|\right)+|\alpha_i| \geq  (i+1)\left(\frac{F-3\varphi}{2}\right)$.
This proves Fact \ref{ind} by induction.

Fact \ref{ind} implies that execution $\alpha_{\lfloor \frac{L}{2} \rfloor-1}$ lasts at least  $(\lfloor \frac{L}{2} \rfloor-1)\left(\frac{F-3\varphi}{2}\right) \in \Omega(EL)$ rounds.
\end{proof}

Our second lower bound shows that no rendezvous algorithm of time complexity of Algorithm {\tt Fast}
can beat the cost complexity of this algorithm.

\begin{theorem}\label{LBcost}
Any deterministic rendezvous algorithm with time $O(E\log L)$ must have cost $\Omega (E\log{L})$.
\end{theorem}

\begin{proof}
Let $\mathcal{A}$ be a rendezvous algorithm such that for every pair of agent labels, and for every pair of starting positions of the agents, rendezvous is completed in at most $cE\log{L}$ rounds, for some constant $c>0$. Our goal is to prove that there exists an execution in which the total combined cost incurred by the agents is in $\Omega(E\log L)$. 

Instead of behaviour vectors of algorithm $\mathcal{A}$, we consider behaviour vectors resulting from procedure {\tt Trim}($\mathcal{A}$). Recall, from the trimming of algorithm $\mathcal{A}$, that $m_x$ is defined to be the maximum value of $|\alpha(x,p_x,y,p_y)|$, taken over all $y \in \{1,\ldots,L\} \setminus \{x\}$ and all nodes $p_x,p_y$. Further, in agent $x$'s trimmed behaviour vector, all entries after $m_x$ have value 0.

Starting with an arbitrary node, label the nodes of the ring using the integers $0,\ldots,n-1$, ascending in the clockwise direction.
This is for analysis only: the agents do not have access to any node labeling.
For simplicity, assume that $n$ is divisible by 6. The proof can be modified in the general case.
Partition the set of nodes into $6$ equal-sized {\em sectors}: for each $j \in \{0,\ldots,5\}$, let $P_j$ be $\{j(\frac{n}{6}),\ldots,(j+1)(\frac{n}{6})-1\}$. For ease of notation, it will be assumed that all subscripts of sectors are taken modulo 6. Next, let $L' = \lceil 6c\log{L}\rceil$, and, for each integer $i \in \{1,\ldots,L'\}$, we define \emph{block} $B_i$ to be the time interval $[(i-1)(\frac{n}{6})+1,\ldots,i(\frac{n}{6})]$. For each agent $x$, let $B(x)$ be the block that contains round $m_x$. Since there are $L$ distinct agent labels and $L' < L$ blocks, it follows from the Pigeonhole Principle that there exist at least $\ell = \lceil L/L' \rceil$ agents $x_1,\ldots,x_\ell$ such that $B(x_1) = \cdots = B(x_\ell)$. Let $M \leq L'$ be the index of the block that contains $m_{x_1},\ldots,m_{x_{\ell}}$. In what follows, we only consider agents from the set $\{x_1,\ldots,x_{\ell}\}$.

Since the number of nodes in a sector is equal to the number of rounds in a block, we observe that the segment of the ring explored by an agent during a single block cannot contain nodes from 3 different sectors. This implies the following fact about which nodes an agent may visit during a given block.

\begin{fact}\label{fact:AtMostOne}
If agent $x$ is located in a sector $P_j$ at the beginning of a block $B_i$, then, in all rounds from the beginning of block $B_i$ until the beginning of block $B_{i+1}$, $x$ is never located at a node outside of $P_{j-1} \cup P_j \cup P_{j+1}$.
\end{fact}

We now define an \emph{aggregate behaviour vector} for $x$, denoted by $Agg_{x,p_x}$, that describes $x$'s movement in execution $\alpha(x,p_x,\bot,\bot)$ during each of the blocks $B_1,\ldots,B_M$. At the beginning of an arbitrary block $B_i$, suppose that agent $x$ is located at a node in $P_j$ for some $j \in \{0,\ldots,5\}$. By Fact \ref{fact:AtMostOne}, at the beginning of block $B_{i+1}$, agent $x$ is located at a node in $P_{j-1} \cup P_{j} \cup P_{j+1}$. For each $i \in \{1,\ldots,M\}$, we define $Agg_{x,p_x}[i]$ to be $z \in\{-1,0,1\}$ if $x$ is located at a node in $P_{j+z}$ at the beginning of block $B_{i+1}$. Note that, for any choice of nodes $p_x,p_x'$ such that $p_x \equiv p'_x(\!\!\!\!\mod \frac{n}{6})$, we get $Agg_{x,p_x} = Agg_{x,p_x'}$. In particular, this implies the following useful fact.

\begin{fact}\label{fact:equiv}
For any agent $y$, $Agg_{y,0} = Agg_{y,\frac{n}{2}}$.
\end{fact}

For any integer-valued vector $V$, define $surplus(V) = \sum_{i=1}^{length(V)} V[i]$.
For a vector $V$, we write $V[a \dots b]$ to denote the part of the vector $V$ between positions $a$ and $b$, inclusive.

The following fact gives a necessary condition on the aggregate vectors of agents that can meet.

\begin{fact}\label{fact:bigdisp}
Consider any distinct agents $x,y$ and any fixed $i,m \in \{1,\ldots, M\}$ such that $i \leq m$. Suppose that, at the beginning of block $B_i$ during the execution $\alpha(x,0,y,\frac{n}{2})$, agent $x$ is located at a node in $P_j$ and agent $y$ is located at a node in $P_{j+3}$. If, for all $k \in \{i,\ldots,m\}$, $|surplus(Agg_{x,0}[i \ldots k])| \leq 1$ and $|surplus(Agg_{y,0}[i \ldots k])| \leq 1$, then agents $x$ and $y$ do not meet in the time interval between the beginning of block $B_i$  and the beginning of block $B_{m+1}$.
\end{fact}
To prove this fact, note that,
since $|surplus(Agg_{x,0}[i \ldots k])| \leq 1$ for all $k \in \{i,\ldots,m\}$, it follows that, in all rounds after the beginning of block $B_i$ until the end of block $B_m$, $x$ is not located at a node outside of $P_{j-1} \cup P_j \cup P_{j+1} = P_{j+5} \cup P_j \cup P_{j+1}$. Next, by Fact \ref{fact:equiv}, we have $Agg_{y,0}=Agg_{y,\frac{n}{2}}$, so $surplus(Agg_{y,0}[i \ldots k]) = surplus(Agg_{y,\frac{n}{2}}[i \ldots k])$ for all $k \in \{i,\ldots,m\}$. Since $|surplus(Agg_{y,0}[i \ldots k])| \leq 1$ for all $k \in \{i,\ldots,m\}$, it follows that, at all times after the beginning of block $B_i$ until the end of block $B_m$, agent $y$ is not located at a node outside of $P_{j+2} \cup P_{j+3} \cup P_{j+4}$. So, during blocks $B_i,\ldots,B_m$ of execution $\alpha(x,0,y,\frac{n}{2})$, agents $x$ and $y$ are never located at the same node. This completes the proof of Fact \ref{fact:bigdisp}.

We now define a \emph{progress vector} for each agent $x$, denoted by $Prog_{x,p_x}$. At a high level, an agent $x$'s progress vector keeps track of each time that $x$ takes a ``significant'' number of steps more in one direction than in the other. Essentially, our goal is to zero out the entries of $x$'s aggregate behaviour vector that amount to $x$ oscillating back and forth on the ring without making sufficient progress towards the other agent. More formally, a node $x$'s progress vector $Prog_{x,p_x}$ is obtained from its aggregate behaviour vector $Agg_{x,p_x}$ in the following way. First, if every prefix of $Agg_{x,p_x}$ has surplus of absolute value at most 1, then $Prog_{x,p_x}$ is defined to be the zero-vector of length $M$. This means that $x$ is essentially idle and waiting for the other agent to come meet it. Otherwise, when there is a prefix of $Agg_{x,p_x}$ that has surplus of absolute value 2, then the smallest such prefix $pre$ is chosen. Next, the `significant' non-zero entries are found, i.e., entries that actually contribute to the large surplus. More formally, consider the case where $surplus(pre) = 2$ (the case where $surplus(pre) = -2$ is symmetric) and suppose that $x$ is initially located at a node in $P_j$. We determine the last block $B_{a}$ during which $x$ moves from $P_j$ to $P_{j+1}$, and, the first block $B_{b}$ during which $x$ moves from $P_{j+1}$ to $P_{j+2}$. Note that, by definition, $b = length(pre)$. Then, we set $Prog_{x,p_x}[i] = Agg_{x,p_x}[i]$ for each $i \in \{a,b\}$, and set $Prog_{x,p_x}[i] = 0$ for each $i \in \{1,\ldots,length(pre)\}\setminus \{a,b\}$. The rest of $Prog_{x,p_x}$ is calculated by repeating the above process on the remaining part of the aggregate behaviour vector, i.e., on $Agg_{x,p_x}[length(pre)+1 \ldots M]$. A complete description is provided in the following pseudocode.

\begin{algorithm}[H]
\caption{\texttt{DefineProgress($Agg$)}}
\begin{algorithmic}[1]
\STATE $Prog \leftarrow \textrm{0-vector of length $M$}$
\STATE $s \leftarrow 1$
\LOOP
\IF{$(s > M)$ OR $|surplus(Agg[s \ldots k])| \leq 1 \textrm{ for all $k \in \{s,\ldots,M\}$}$}
	\STATE \% {\it Case 1: no surplus with absolute value at least 2}
	\STATE \% {\it We don't preserve any remaining entries from $Agg$}
	\RETURN $Prog$
\ELSE
	\STATE \% {\it Case 2: there exists a prefix such that surplus has absolute value 2}
	\STATE \% {\it Find the 2 ``significant'' entries to preserve from $Agg$}
	\STATE \begin{varwidth}[t]{\linewidth}$b \leftarrow \textrm{ smallest $i \geq s$ such that}$\par
		 \hskip\algorithmicindent\hskip\algorithmicindent $|surplus(Agg[s \ldots i])| = 2$
		 \end{varwidth}\label{firstassign}
	\STATE \begin{varwidth}[t]{\linewidth} $a \leftarrow \textrm{ smallest integer in $\{s,\ldots,b\}$ such that,}$\par
		\hskip\algorithmicindent\hskip\algorithmicindent for all $i \in \{a,\ldots,b\}$, $|surplus(Agg[s \ldots i])| \geq 1$
		\end{varwidth}\label{secondassign}
	\STATE set $Prog[a]$ and $Prog[b]$ equal to $Agg[b]$\label{valueset}
	\STATE $s \leftarrow b+1$
\ENDIF
\ENDLOOP
\end{algorithmic}
\end{algorithm}

In the construction of $Prog_{x,0}$, consider an arbitrary iteration $j$ of the loop. We denote by $s_j$ the value of $s$ at the beginning of iteration $j$, and we denote by $a_j$ and $b_j$ the values of $a$ and $b$, respectively, at the end of iteration $j$. In what follows, we will use the following invariants about the construction of $Prog_{x,0}$.

\begin{fact}\label{fact:inequalities}
For an arbitrary loop iteration $j$ before the final one, we have $s_j \leq a_j < b_j < s_{j+1}$.
\end{fact}
To see why this is true, we first note that $a_j$ is chosen from the range $\{s_j,\ldots,b_j\}$. It cannot be the case that $a_j=b_j$, since $|surplus(Agg[s_j \ldots b_j])| = 2 > 1=|surplus(Agg[s_j \ldots a_j])|$. The last inequality holds since, at the end of the loop, $s_{j+1}$ is set to $b_j+1$.

\begin{fact}\label{fact:PreserveNonzero}
At line \ref{valueset}, $Agg[a] = Agg[b] = Prog[b] = Prog[a] \neq 0$.
\end{fact}
To see why this is true, it is sufficient to consider the case where $surplus(Agg[s \ldots b]) > 0$ and prove that $Agg[a]=Agg[b]=1$ (in the case where this surplus is negative, a similar proof shows that $Agg[a]=Agg[b]=-1$.) From line \ref{firstassign}, $b$ is the smallest index greater than or equal to $s$ such that $surplus(Agg[s \ldots b]) = 2$. Clearly, $b > s$ since, otherwise, $surplus(Agg[s \ldots b]) = Agg[s] \in \{-1,0,1\}$. Further, if $Agg[b] \in \{-1,0\}$, then $surplus(Agg[s \ldots b-1]) \geq surplus(Agg[s \ldots b-1]) + Agg[b] = surplus(Agg[s \ldots b]) = 2$, which contradicts the minimality of $b$. So, we conclude that $Agg[b]=1$. Next, from line \ref{secondassign} and the fact that $surplus(Agg[s \ldots b]) = 2$, $a$ is the smallest index in the range $\{s,\ldots,b\}$ such that, for all $i \in \{a,\ldots,b\}$, $surplus(Agg[s \ldots i]) \geq 1$. If $a = s$, then $Agg[a] = surplus(Agg[s \ldots a]) \geq 1$, which implies that $Agg[a]=1$. If $a > s$ and $Agg[a] \in \{-1,0\}$, then $surplus(Agg[s \ldots a-1]) \geq surplus(Agg[s \ldots a-1])+Agg[a]=surplus(Agg[s \ldots a]) \geq 1$, which contradicts the minimality of $a$. So, we conclude that $Agg[a]=1$, which completes the proof of Fact \ref{fact:PreserveNonzero}.


Our next goal is to show that progress vectors of different agents must be distinct. This is not immediately clear because, in the construction of progress vectors, distinct aggregate behaviour vectors can be mapped to equal progress vectors. The following technical result will be used to show that the entries of an agent's aggregate behaviour vector that got converted to zeroes in the agent's progress vector actually do not contribute to the completion of rendezvous.

\begin{fact}\label{fact:zeros}
Consider any agent $x$, and consider any integers $i_1\leq i_2$ in $\{1,\ldots,M\}$ such that $Prog_{x,0}[i_1 \ldots i_2]$ is a maximal sequence of 0's in $Prog_{x,0}$. Then,
\begin{enumerate}
\item for each $i \in \{i_1,\ldots,i_2\}$, $|surplus(Agg_{x,0}[i_1 \ldots i])| \leq 1$, and,
\item if $i_2 \neq M$, $surplus(Agg_{x,0}[i_1 \ldots i_2]) = 0$.
\end{enumerate}
\end{fact}
To prove this fact, consider any $i_1 \leq i_2$ in $\{1,\ldots,M\}$ such that $Prog_{x,0}[i_1 \ldots i_2]$ is a maximal sequence of 0's in $Prog_{x,0}$. In the construction of $Prog_{x,0}$, there exists an iteration $j$ such that either:
\begin{enumerate}
\item $i_1 = s_j, i_2 = a_j-1$, or,
\item $i_1 = a_j+1, i_2 = b_j-1$, or,
\item $i_1 = s_j, i_2 = M$.
\end{enumerate}
If $i_1 = s_j$ and $i_2 = a_j-1$, Fact \ref{fact:inequalities} implies that $i_1 \leq i_2 < b_j$. So, by the minimality of $b_j$, for each $i \in \{i_1,\ldots,i_2\}$, we have $|surplus(Agg_{x,0}[i_1 \ldots i])| \leq 1$. Also, by the minimality of $a_j$, $|surplus(Agg_{x,0}[s_j \ldots a_j-1])| < 1$, that is, $surplus(Agg_{x,0}[i_1 \ldots i_2]) = 0$.

If $i_1 = a_j+1$ and $i_2 = b_j-1$, note that, by the choice of $a_j$, $|surplus(Agg_{x,0}[s_j \ldots i])| \geq 1$ for all $i \in \{i_1-1,\ldots i_2\}$. Also, by the minimality of $b_j$, $|surplus(Agg_{x,0}[s_j \ldots i])| \leq 1$ for all $i \in \{i_1-1,\ldots,i_2\}$. Therefore, for all $i \in \{i_1-1,\ldots,i_2\}$, we have $|surplus(Agg_{x,0}[s_j \ldots i])| = 1$. So, for an arbitrary $i \in \{i_1,\ldots,i_2\}$, $|surplus(Agg_{x,0}[s_j \ldots i])| = 1$ and  $|surplus(Agg_{x,0}[s_j \ldots i-1])|=1$, which implies that $Agg_{x,0}[i] \in \{-2,0,2\}$. We conclude that $Agg_{x,0}[i]=0$ for all $i \in \{i_1,\ldots,i_2\}$. It follows that $surplus(Agg_{x,0}[i_1 \ldots i]) = 0$ for all $i \in \{i_1,\ldots,i_2\}$.

If $i_1 = s_j$ and $i_2=M$, then we must have reached Case 1 in loop iteration $j$. It follows that $|surplus(Agg_{x,0}[i_1 \ldots i])| \leq 1$ for each $i \in \{i_1,\ldots,i_2\}$. This completes the proof of Fact \ref{fact:zeros}.

We now show that, in order to meet in every execution, agents must have distinct progress vectors.

\begin{fact}\label{fact:distinct}
For any distinct agents $x,y$, if $Prog_{x,0} = Prog_{y,0}$, then $x$ and $y$ do not meet in execution $\alpha(x,0,y,\frac{n}{2})$.
\end{fact} 
To establish this fact, it is sufficient to prove the following statement: 
\begin{quotation}
\noindent
for all $i \in \{1,\ldots,M\}$, if
\begin{itemize}
\item $i=1$ or $Prog_{x,0}[i-1] \neq 0$, and,
\item
at the beginning of a block $B_i$ of execution $\alpha(x,0,y,\frac{n}{2})$, for some $j \in \{0,\ldots,5\}$, $x$ is at a node in $P_j$ and $y$ is at a node in $P_{j+3}$, and,
\item
$Prog_{x,0}[i \ldots M] = Prog_{y,0}[i \ldots M]$, 
\end{itemize}
then $x$ and $y$ do not meet after the beginning of block $B_i$ of execution $\alpha(x,0,y,\frac{n}{2})$. 
\end{quotation}

We prove this statement by induction on the number $k$ of non-zero entries in $Prog_{x,0}[i \ldots M]$, for arbitrary $i \in \{1,\ldots,M\}$. The base case of the induction is for $k=0$. For an arbitrary $i \in \{1,\ldots,M\}$, suppose that the three conditions of the statement hold. Then, $Prog_{x,0}[i \ldots M]$ and $Prog_{y,0}[i \ldots M]$ are sequences of consecutive 0's in $Prog_{x,0}$ and $Prog_{y,0}$, respectively. Since $i=1$ or $Prog_{x,0}[i-1] \neq 0$, these sequences are maximal. So, by Fact \ref{fact:zeros}, every prefix of $Agg_{x,0}[i \ldots M]$ and every prefix of $Agg_{y,0}[i \ldots M]$ have surpluses with absolute value at most 1. By Fact \ref{fact:bigdisp}, $x$ and $y$ do not meet in execution $\alpha(x,0,y,\frac{n}{2})$ after the beginning of block $B_i$.

As induction hypothesis, assume that, for all $i \in \{1,\ldots,M\}$, if the three conditions of the statement hold, and, for some $k \geq 0$, there are $k$ non-zero entries in $Prog_{x,0}[i \ldots M]$, then $x$ and $y$ do not meet after the beginning of block $B_i$ of execution $\alpha(x,0,y,\frac{n}{2})$. 

Now, consider an arbitrary $i \in \{1,\ldots,M\}$. Suppose that there are $k+1$ non-zero entries in $Prog_{x,0}[i \ldots M]$, and the three conditions of the statement hold.
Let $i'$ be the first non-zero entry in $Prog_{x,0}[i \ldots M]$. We set out to show that no rendezvous occurs during blocks $B_i,\ldots,B_{i'}$ and that the three conditions of the statement hold when $i$ is replaced with $i'+1$. This is sufficient to complete the proof: since the number of non-zero entries in $Prog_{x,0}[i'+1 \ldots M]$ is $k$, the induction hypothesis implies that agents $x$ and $y$ do not meet after the beginning of block $B_{i'+1}$.

First, we show that rendezvous does not occur during blocks $B_i,\ldots,B_{i'-1}$. If $i=i'$, there is nothing to prove. Otherwise, since $i=1$ or $Prog_{x,0}[i-1] \neq 0$, it follows that $Prog_{x,0}[i\ldots i'-1]$ is a maximal sequence of 0's. Therefore, by Fact \ref{fact:zeros}, every prefix of $Agg_{x,0}[i \ldots i'-1]$ and every prefix of $Agg_{y,0}[i \ldots i'-1]$ have surpluses with absolute value at most 1. By Fact \ref{fact:bigdisp}, $x$ and $y$ do not meet during any of the blocks $B_i,\ldots,B_{i'-1}$. 

Next, we show that, at the beginning of block $B_{i'}$, $x$ is located at a node in $P_j$ and that $y$ is located at a node in $P_{j+3}$. If $i=i'$, this is true by assumption. Otherwise, note that $Prog_{x,0}[i\ldots i'-1]$ is a maximal sequence of 0's and that $i'-1 < i' \leq M$. Therefore, by Fact \ref{fact:zeros}, $surplus(Agg_{x,0}[i\ldots i'-1]) = 0$, and hence, at the beginning of block $B_{i'}$ agent $x$ is in the same sector as at the beginning of block $B_i$. The same holds for agent $y$. 
We conclude that rendezvous does not occur during block $B_{i'}$. This follows from Fact \ref{fact:bigdisp}, since $|surplus(Prog_{x,0}[i'\ldots i'])| \leq 1$ and $|surplus(Prog_{x,0}[i'\ldots i'])| \leq 1$.

Finally, we show that the three conditions of the statement hold at the beginning of block $B_{i'+1}$. The first condition holds since $Prog_{x,0}[i'] \neq 0$. Also, the third condition holds since we assumed that $Prog_{x,0}[i \ldots M] = Prog_{y,0}[i \ldots M]$. To show that the second condition holds, note that, by the definition of the aggregate behaviour vector, at the beginning of block $B_{i'+1}$, agent $x$ is located at a node in $P_{j+Agg_{x,0}[i']}$, and agent $y$ is located at a node in $P_{j+3+Agg_{y,\frac{n}{2}}[i']}$. By Facts \ref{fact:equiv} and \ref{fact:PreserveNonzero}, $Agg_{x,0}[i'] = Prog_{x,0}[i'] = Prog_{y,0}[i'] = Agg_{y,0}[i'] = Agg_{y,\frac{n}{2}}[i']$.  Thus, for $j' = j+Agg_{x,0}[i']$, agent $x$ is located at a node in $P_{j'}$ and $y$ is located at a node in $P_{j'+3}$ at the beginning of block $B_{i'+1}$. This completes the proof by induction and hence completes the proof of Fact \ref{fact:distinct}.

Using the fact that the progress vectors must all be distinct (cf. Fact \ref{fact:distinct}), we now show that there must be a progress vector of large weight.

\begin{fact}\label{fact:ExistsHeavy}
Consider the $\ell = \lceil \frac{L}{\lceil 6c\log{L}\rceil} \rceil$ distinct progress vectors $Prog_{x_1},\ldots,Prog_{x_\ell}$. There exists $j \in \{1,\ldots,\ell\}$ such that $Prog_{x_j,0}$ contains $\Omega(\log{L})$ non-zero entries.
\end{fact}
To prove this fact, we show that, for a sufficiently small constant $\gamma$, there are fewer than $\ell$ distinct vectors of length $M$ with at most $\gamma\log{L}$ non-zero entries. The fact will then follow from the Pigeonhole Principle.

Using the bound $\binom{n}{k} \leq \left( \frac{en}{k} \right)^k$ (where $e$ is the Euler constant), the total number of vectors of length $n$ with at most $k$ non-zero entries can be bounded above as follows:
$$ \binom{n}{0} + \binom{n}{1} + \cdots + \binom{n}{k}  \leq  (k+1)e^k\left(\frac{n}{k}\right)^k \leq (2^k)e^k\left(\frac{n}{k}\right)^k.$$

Let $f$ be a constant for which $f \geq 1/(4e)$ and $ \lceil 6c\log{L}\rceil \leq f\log L$.
Substituting $n = M \leq f \log{L}$ and $k =\lfloor \gamma\log{L}\rfloor$, we get that the number of distinct vectors of length $M$ with at most $\gamma\log{L}$ non-zero entries is bounded above by $(2e)^{\lfloor \gamma\log{L}\rfloor}\left(\frac{f\log{L}}{\lfloor \gamma\log{L}\rfloor}\right)^{\lfloor \gamma\log{L}\rfloor} \leq  \left(4ef\right)^{\gamma\log{L}}\left(\left(\frac{1}{\gamma}\right)^{\gamma}\right)^{\log{L}}$. 
Next, it is not difficult to show that $\left(\frac{1}{\gamma}\right)^{\gamma}$ converges to 1 as $\gamma$ approaches 0. 
Also since $\ell = \lceil \frac{L}{\lceil 6c\log{L} \rceil} \rceil$, there exists a positive constant $\beta < 1$ such that $\ell > L^{\beta}$. 
So, we pick sufficiently small $1>\gamma'>0$ such that $\log{\left(\left(\frac{1}{\gamma'}\right)^{\gamma'}\right)} < \beta/2$, and, for all $\gamma \leq \gamma'$, $\left(\left(\frac{1}{\gamma}\right)^{\gamma}\right)^{\log{L}} < L^{\beta/2}$. Next, let $\gamma'' = \frac{\beta}{2\log{(4ef)}}$, and note that, for all $\gamma \leq \gamma''$, $\left(4ef\right)^{\gamma''\log{L}} \leq L^{\beta/2}$. Therefore, taking $\gamma = \min\{\gamma',\gamma''\}$, it follows that the number of distinct vectors of length $M$ with at most $\gamma\log{L}$ non-zero entries it at most $\left(4ef\right)^{\gamma\log{L}}\left(\left(\frac{1}{\gamma}\right)^{\gamma}\right)^{\log{L}} \leq L^\beta$, which is less than $\ell$. This completes the proof of Fact \ref{fact:ExistsHeavy}.
 
We now set out to prove that there exists an agent incurring cost $\Omega(E\log L)$ in some execution of the algorithm. During each iteration $i$ of the loop in the construction of $Prog_{x_j,0}$ (except for the last), two entries, at positions $a_i$ and $b_i$, are set to non-zero values. In particular, this means that, for all $d \in \{a_i+1,\ldots,b_i-1\}$, $Prog_{x_j,0}[d]=0$. Let $k$ be the number of iterations  in which two entries are set to non-zero values. From Fact \ref{fact:inequalities}, we know that $a_1 < b_1 < a_2 < b_2 < \cdots < a_k < b_k$. 
From Fact \ref{fact:PreserveNonzero}, we know that for each $i \in \{1,\dots , k\}$ we have $Prog_{x_j,0}[a_i] = Prog_{x_j,0}[b_i] \neq 0$.

The following fact shows that the number of non-zero entries in a progress vector induces a lower bound on the cost incurred by an agent. 

\begin{fact}\label{fact:costly}
Consider any agent $x$ and any integers $a_1,b_1,\ldots,a_k,b_k \in \{1,\ldots,M\}$ such that
\begin{itemize}
\item $a_1<b_1<\cdots<a_k<b_k$, and,
\item for each $i \in \{1,\ldots,k\}$, $Prog_{x,0}[a_i] = Prog_{x,0}[b_i] \neq 0$, and,
\item for each $i \in \{1,\ldots,k\}$ and each $d \in \{a_i+1,\ldots,b_i-1\}$, $Prog_{x,0}[d]=0$.
\end{itemize}
During execution $\alpha(x,0,\bot,\bot)$, agent $x$ performs at least $\frac{kE}{6}$ edge traversals.
\end{fact}
To see why this is true, consider an arbitrary $i \in \{1,\ldots,k\}$ and suppose that $Prog_{x,0}[a_i] = Prog_{x,0}[b_i] = 1$ (the case where $Prog_{x,0}[a_i] = Prog_{x,0}[b_i]=-1$ is symmetric). At the beginning of block $B_{a_i}$, agent $x$ is located in some sector $P_j$. By Fact \ref{fact:PreserveNonzero}, $Agg_{x,0}[a_i] = Prog_{x,0}[a_i]$, so, at the beginning of block $B_{a_i+1}$, agent $x$ is located in sector $P_{j+1}$. Next, since $Prog_{x,0}[a_i+1 \ldots b_i-1]$ is a maximal sequence of 0's, and $b_i-1 < M$, it follows from Fact \ref{fact:zeros} that $surplus(Agg_{x,0}[a_i+1 \ldots b_i-1]) = 0$. Therefore, at the beginning of block $B_{b_i}$, $x$ is still located in sector $P_{j+1}$. Finally, by Fact \ref{fact:PreserveNonzero}, $Agg_{x,0}[b_i] = Prog_{x,0}[b_i]$, so, at the beginning of block $B_{b_i+1}$, $x$ is located in sector $P_{j+2}$. It follows that, from the beginning of block $B_{a_i}$ until the end of block $B_{b_i}$, agent $x$ must have visited every node in sector $P_{j+1}$, i.e., it traversed at least $\frac{E}{6}$ edges. The inequalities $a_1 < b_1 < \cdots < a_k < b_k$ give us $k$ disjoint time intervals during each of which at least $\frac{E}{6}$ edges are traversed.
This completes the proof of Fact \ref{fact:costly}.
%
%

By Fact \ref{fact:ExistsHeavy} there exists an agent $x_j$ such that $Prog_{x_j,0}$ has at least $\Omega(\log L)$ non-zero entries. 
Applying Fact \ref{fact:costly} to this agent implies that it incurs cost  $\Omega(E\log L)$ in its solo execution of the trimmed version of algorithm $\cal A$.
Hence, there exists an agent $y$ and nodes $p_{x_j}$ and $p_y$, such that agent $x_j$ incurs the same cost in execution
$\alpha(x_j,p_{x_j},y,p_y)$. This completes the proof of  Theorem \ref{LBcost}.
\end{proof}

\section{Conclusion}

We established tight tradeoffs at both ends of the time/cost tradeoff curve,
up to multiplicative constants. This suggests that if we want to minimize cost (respectively time) of rendezvous, then our natural algorithms {\tt Cheap}
(respectively {\tt Fast}) are good choices.  
A challenging
open problem yielded by our work is establishing the entire precise tradeoff curve, i.e., finding, for each cost value between $\Theta(E)$ and $\Theta(E\log L)$, the minimum time of rendezvous that can be performed at this cost. In particular,  it is natural to ask if the performance of our Algorithm {\tt FastWithRelabeling}$(s)$ is on, or close to, this
optimal tradeoff curve. 

In this paper, we adopted a model in which both agents are located at their starting positions from the beginning, and the adversary wakes them up possibly 
at different times. Hence, if the delay is sufficiently large, it is possible that the earlier agent finds the later agent before it even starts executing the algorithm. Consequently, both time and cost are counted from the wake-up of the earlier agent. 
Such an approach is natural, since we are interested in both the time and cost of the algorithm, and the cost is defined as 
the combined number of edge traversals by both agents. An alternative model, used in papers dealing only with time of rendezvous (cf. \cite{DFKP,TSZ07}),
assumes that agents are ``parachuted'' onto their respective starting positions at the time of their wake-up. In this model, the earlier agent cannot find the later
agent before its wake-up because the later agent is not yet present. Hence, in  \cite{DFKP,TSZ07}, time was counted from the wake-up of the {\em later} agent, since otherwise rendezvous time can be made arbitrarily large by an adversary. Similarly, in our case, we would have to count both the time and the cost since the wake-up of the later agent.
This does not seem natural as far as cost is concerned, because 
it is often the case that incurring cost results in consuming a limited resource, such as energy. So ignoring the cost incurred by the earlier agent until the wake-up
of the later agent is unrealistic. 
Nevertheless, the time and cost
complexities of our algorithms do not change in this alternative model (although the proofs have to be slightly modified). Our lower bounds are not affected either,
as they work even for simultaneous start.

Finally, we address our assumption that an exploration procedure and its cost $E$ are known. As we argued in the introduction, the exploration time
is a benchmark for the cost of rendezvous. Further, this knowledge can be deduced by the agents from an upper bound on the size of the graph.
What if agents do not have any such upper bound? It turns out that our algorithms can be slightly modified to preserve their time and cost complexities in this case as well. Recall that a Universal Exploration Sequence (UXS) is a sequence of integers 
that can be used to explore any graph of size at most $m$ at cost $R(m)$, for some fixed polynomial $R$, starting at any node of the graph. 
Let  ${\tt EXPLORE}_i$ be the the UXS-based exploration procedure for the class of graphs of size at most $2^i$, and let $E_i$ be the time of  ${\tt EXPLORE}_i$.
Each of our algorithms can be modified by iterating the original algorithm using ${\tt EXPLORE}={\tt EXPLORE}_i$ and $E=E_i$ in the $i$-th iteration. Iterations proceed until rendezvous, which will occur when $2^i$ is at least the actual size of the graph. Due to telescoping, the
time and cost complexities will not change. 

%
%

\bibliographystyle{plain}


\end{document}